\documentclass[12pt]{article}

\usepackage{amsthm}
\usepackage{amssymb}
\usepackage{amsmath}
\usepackage{amsfonts}
\usepackage{times}
\usepackage{bbm}
\usepackage{cite}
\usepackage{color, soul}
\usepackage{hyperref}
\hypersetup{
     colorlinks   = true,
     citecolor    = black,
     linkcolor    = blue
}

\usepackage[paper=a4paper, top=2.5cm, bottom=2.5cm, left=2.5cm, right=2.5cm]{geometry}

\def\qed{\hfill $\vcenter{\hrule height .3mm
\hbox {\vrule width .3mm height 2.1mm \kern 2mm \vrule width .3mm
height 2.1mm} \hrule height .3mm}$ \bigskip}

\def \RR {\mathbb R}

\def \EE {\mathbb E}

\def \eps {\varepsilon}

\def \FF {\mathcal{F}}

\newcommand\norm[1]{\left\lVert#1\right\rVert}
\newtheorem{theorem}{Theorem}
\newtheorem{lemma}{Lemma}

\newtheorem{proposition}{Proposition}
\newtheorem{corollary}[theorem]{Corollary}
\theoremstyle{definition}

\theoremstyle{remark}
\newtheorem{remark}[theorem]{Remark}

\title{Stability of the Shannon-Stam inequality via the F\"ollmer process}

\author{Ronen Eldan\thanks{Weizmann Institute of Science. Incumbent of the Elaine Blond career development chair. Supported by a European Research Council Starting Grant (ERC StG) and by an Israel Science Foundation grant no. 715/16.} \and Dan Mikulincer\thanks{Weizmann Institute of Science. Supported by an Azrieli foundation fellowship.}}
 
\begin{document}

\maketitle

\begin{abstract}
	We prove stability estimates for the Shannon-Stam inequality (also known as the entropy-power inequality) for log-concave random vectors in terms of entropy and transportation distance. In particular, we give the first stability estimate for general log-concave random vectors in the following form: for log-concave random vectors $X,Y \in \RR^d$, the deficit in the Shannon-Stam inequality is bounded from below by the expression
	$$
	C \left(\mathrm{D}\left(X||G\right) + \mathrm{D}\left(Y||G\right)\right),
	$$
	where $\mathrm{D}\left( \cdot ~ ||G\right)$ denotes the relative entropy with respect to the standard Gaussian and the constant $C$ depends only on the covariance structures and the spectral gaps of $X$ and $Y$. In the case of uniformly log-concave vectors our analysis gives dimension-free bounds. Our proofs are based on a new approach which uses an entropy-minimizing process from stochastic control theory.
\end{abstract}

	\section{Introduction}
	Let $\mu$ be a probability measure on $\RR^d$ and $X \sim \mu$. Denote by $\mathrm{h}(\mu)$, the differential entropy of $\mu$ which is defined to be
	$$\mathrm{h}(\mu): = \mathrm{h}(X) = -\int\limits_{\RR^d}\ln\left(\frac{d\mu}{dx}\right)d\mu.$$ 
	One of the fundamental results of information theory is the celebrated Shannon-Stam inequality which asserts that for independent vectors $X$, $Y$ and $\lambda \in (0,1)$
	\begin{equation} \label{eq: shannon-stam}
	 \mathrm{h}\left(\sqrt{\lambda}X + \sqrt{1-\lambda}Y\right) \geq \lambda \mathrm{h}(X) + (1-\lambda)\mathrm{h}(Y).
	\end{equation}
	We remark that Stam \cite{stam1959some} actually proved the equivalent statement
	\begin{equation} \label{eq: EPI}
		e^{\frac{2\mathrm{h}(X+Y)}{d}} \geq e^{\frac{2\mathrm{h}(X)}{d}} + e^{\frac{2\mathrm{h}(Y)}{d}},
	\end{equation}
	first observed by Shannon in \cite{shannon1948mathematical}, and known today as the entropy power inequality.
	To state yet another equivalent form of the inequality, for any positive-definite matrix, $\Sigma$, we set $\gamma_\Sigma$ as the  centered Gaussian measure on $\RR^d$ with density
	$$\frac{d\gamma_{\Sigma}(x)}{dx} = \frac{e^{-\frac{\langle x, \Sigma^{-1}x\rangle}{2}}}{\sqrt{\det(2\pi\Sigma)}}.$$
	For the case where the covariance matrix is the identity, $\mathrm{I}_d$, we will also write $\gamma := \gamma_{\mathrm{I}_d}$.
	If $Y \sim \nu$ we set the relative entropy of $X$ with respect to $Y$ as
	$$\mathrm{D}(\mu||\nu): = \mathrm{D}(X||Y) = \int\limits_{\RR^d}\ln \left(\frac{d\mu}{d\nu}\right)d\mu.$$
	For $G \sim \gamma$, the differential entropy is related to the relative entropy by
	\begin{align*}
	\mathrm{D}(X ||G) &= -\mathrm{h}(X) - \frac{1}{2}\EE\left[\norm{X}_2^2\right]+\frac{d}{2}\ln(2\pi).
	\end{align*}
	Thus, when $X$ and $Y$ are independent and centered the statement 
\begin{equation} \label{eq: relative shannon-stam}
	\mathrm{D}\left(\sqrt{\lambda}X + \sqrt{1-\lambda}Y\big|\big|G\right) \leq \lambda \mathrm{D}(X||G) + (1 -\lambda)\mathrm{D}(Y||G),
\end{equation}
	is equivalent to \eqref{eq: shannon-stam}.
	Shannon noted that in the case that $X$ and $Y$ are Gaussians with proportional covariance matrices, both sides of \eqref{eq: EPI} are equal. Later, in \cite{stam1959some} it was shown that this is actually a necessary condition for the equality case. We define the deficit in \eqref{eq: relative shannon-stam} as 
	$$\delta_{EPI, \lambda}(\mu,\nu):=\delta_{EPI, \lambda}(X,Y)=  \Bigl (\lambda \mathrm{D}(X||G) + (1 -\lambda)\mathrm{D}(Y||G)\Bigr) - \mathrm{D}\left(\sqrt{\lambda}X + \sqrt{1-\lambda}Y\big|\big|G\right),$$
	and are led to the question: \emph{what can be said about $X$ and $Y$ when $\delta_{EPI, \lambda}(X,Y)$ is small?}
	One might expect that, in light of the equality cases, a small deficit in \eqref{eq: relative shannon-stam} should imply that $X$ and $Y$ are both close, in some sense, to a Gaussian. A recent line of works has focused on an attempt to make this intuition precise (see e.g., \cite{toscani2015strengthened, courtade2016wasserstein}), which is also our main goal in the present work. In particular, we give the first stability estimate in terms of relative entropy.
	A good starting point is the work of Courtade, Fathi and Pananjady (\cite{courtade2016wasserstein}) which considers stability in terms of the Wasserstein distance (also known as quadratic transportation). The Wasserstein distance is defined by
	$$\mathcal{W}_2(\mu,\nu) =\inf\limits_{\pi}\sqrt{\int\limits_{\RR^{2d}}\norm{x-y}_2^2d\pi(x,y)},$$
	where the infimum is taken over all couplings $\pi$ whose marginal laws are $\mu$ and $\nu$. A crucial observation made in their work is that without further assumptions on the measures $\mu$ and $\nu$, one should not expect meaningful stability results to hold. Indeed, for any $\lambda \in (0,1)$ they show that there exists a family of measures $\{\mu_\eps\}_{\eps > 0}$ such that  $\delta_{EPI, \lambda}(\mu_\eps,\mu_\eps) < \eps$ and such that for any Gaussian measure $\gamma_\Sigma$, $\mathcal{W}_2(\mu_\eps, \gamma_\Sigma) \geq \frac{1}{3}$. Moreover, one may take $\mu_\eps$ to be a mixture of Gaussians.	
	Thus, in order to derive quantitative bounds it is necessary to consider a more restricted class of measures. We focus on the class of log-concave measures which, as our method demonstrates, turns out to be natural in this context. 
	\subsection*{Our Contribution} \label{sec: results}
A measure is called log-concave  if it is supported on some subspace of $\RR^d$ and, relative to the Lebesgue measure of that subspace, it has a density $f$ for which
$$-\nabla^2 \ln(f(x)) \succeq 0 \text{ for all } x,$$
where $\nabla^2$ denotes the Hessian matrix, and we consider the inequality in the sense of positive definite matrices.
Our first result will rely on a slightly stronger condition known as \emph{uniform} log-concavity.
If there exists $\xi > 0$ such that
$$-\nabla^2 \ln(f(x)) \succeq {\xi}\mathrm{I}_d \text{ for all } x,$$
then we say that the measure is $\xi$-uniformly log-concave.
	\begin{theorem} \label{thm: stability for uniform}
		Let $X$ and $Y$ be $1$-uniformly log-concave centered vectors, and denote by $\sigma^2_X,\sigma^2_Y$ the respective minimal eigenvalues of their covariance matrices. Then there exist Gaussian vectors $G_X$ and $G_Y$ such that for any $\lambda \in (0,1)$,
			{\small\begin{align*}
			\delta_{EPI, \lambda}(X,Y) \geq
			\frac{\lambda(1-\lambda)}{2}\left(\sigma_X^4\mathrm{D}\left(X||G_X\right) + \sigma_Y^4\mathrm{D}\left(Y||G_Y\right)+ \frac{\sigma_X^4}{2}\mathrm{D}\left(G_X||G_Y\right) + \frac{\sigma_Y^4}{2}\mathrm{D}\left(G_Y||G_X\right)\right).
			\end{align*}}
	\end{theorem}
To compare this with the main result of \cite{courtade2016wasserstein} we recall the transportation-entropy inequality due to Talagrand  (\cite{talagrand1996transportation}) which states that $$\mathcal{W}_2^2(X,G) \leq 2\mathrm{D}(X||G).$$
As a conclusion we get 
\begin{align*}
	\delta_{EPI, \lambda}(X,Y) \geq C_{\sigma_X,\sigma_Y}
	\frac{\lambda(1-\lambda)}{2}\left(\mathcal{W}_2^2\left(X,G_X\right) + \mathcal{W}_2^2\left(Y,G_Y\right)+ \mathcal{W}_2^2\left(G_X,G_Y\right)\right),
\end{align*}
where $C_{\sigma_X,\sigma_Y}$ depends only on $\sigma_X$ and $\sigma_Y$.
Up to this constant, this is precisely the main result of \cite{courtade2016wasserstein}. In fact, our method can reproduce their exact result, which we present as a warm up in the next section.
We remark that as the underlying inequality is of information-theoretic nature, it is natural to expect that stability estimates are expressed in terms of relative entropy.

A random vector is isotropic if it is centered and its covariance matrix is the identity. By a re-scaling argument the above theorem can be restated for uniform log-concave isotropic random vectors. 
\begin{corollary} \label{cor: EPI for xi log-concave}
	Let $X$ and $Y$ be $\xi$-uniformly log-concave and isotropic random vectors, then there exist Gaussian vectors $G_X$ and $G_Y$  such that for any $\lambda \in (0,1)$ 
		\begin{align*}
	\delta_{EPI, \lambda}(X,Y) \geq
	\frac{\lambda(1-\lambda)}{2}\xi^2\left(\mathrm{D}\left(X||G_X\right) + \mathrm{D}\left(Y||G_Y\right)+ \frac{1}{2}\mathrm{D}\left(G_X||G_Y\right) + \frac{1}{2}\mathrm{D}\left(G_Y||G_X\right)\right).
	\end{align*}
\end{corollary}
In our estimate for general log-concave vectors, the dependence on the parameter $\xi$ will be replaced by the spectral gap of the measures. We say that a random vector $X$ satisfies a Poincar\'e inequality if there exists a constant $C>0$ such that
$$\EE\left[\mathrm{Var}(\psi(X))\right]\leq C \EE\left[\norm{\nabla\psi(X)}_2^2\right], \text{ for all test functions } \psi.$$
We define $C_p(X)$ to be the smallest number such that the above equation holds with $C=C_p(X)$, and refer to this quantity as the Poincar\'e constant of $X$. The inverse quantity, $C_p(X)^{-1}$ is referred to as the \emph{spectral gap} of $X$.
\begin{theorem} \label{thm: EPI for general log-concave}
	Let $X$ and $Y$ be centered log-concave vectors with $\sigma^2_X$, $\sigma_Y^2$ denoting the minimal eigenvalues of their covariance matrices. Assume that 
	$\mathrm{Cov}(X) +  \mathrm{Cov}(Y) =2\mathrm{I}_d$ and set $\max\left(\frac{\mathrm{C_p}(X)}{\sigma_X^2},\frac{\mathrm{C_p}(Y)}{\sigma^2_Y}\right) = \mathrm{C_p}$. Then, if $G$ denotes the standard Gaussian, for every $\lambda \in (0,1)$
	\begin{align*}
	&\delta_{EPI, \lambda}(X,Y) \geq 
	K\lambda(1-\lambda)\left(\frac{\min(\sigma^2_Y,\sigma_X^2)}{\mathrm{C_p}}\right)^3\left(\mathrm{D}\left(X||G\right) + \mathrm{D}\left(Y||G\right)\right),
	\end{align*}
	where $K >0$ is a numerical constant, which can be made explicit.
\end{theorem}
\begin{remark}
 For $\xi$-uniformly log-concave vectors, we have the relation, $\mathrm{C_p}(X) \leq \frac{1}{\xi}$ (this is a consequence of the Brascamp-Lieb inequality \cite{brascamp1976extensions}, for instance). Thus, considering Corollary \ref{cor: EPI for xi log-concave}, one might have expected that the term $\mathrm{C^3_p}$ could have been replaced by $\mathrm{C^2_p}$ in Theorem \ref{thm: EPI for general log-concave}. We do not know if either result is tight.
\end{remark}
\begin{remark} \label{rem: KLS}
	Bounding the Poincar\'e constant of an isotropic log-concave measure is the object of the long standing Kannan-Lov\'az-Simonovits (KLS) conjecture (see \cite{kannan1995isoperimetric,ledoux2004spectral} for more information). The conjecture asserts that there exists a constant $K >0$, independent of the dimension, such that for any isotropic log-concave vector $X$, $\mathrm{C_p}(X) \leq K$. The best known bound is due to Lee and Vempala which showed in \cite{lee2017eldan} that if $X$ is a a $d$-dimensional log-concave vector, $\mathrm{C_p}(X) = O \left(\sqrt{d}\right).$
\end{remark}
Concerning the assumptions of Theorem \ref{thm: EPI for general log-concave}; note that as the EPI is invariant to linear transformation, there is no loss in generality in assuming $\mathrm{Cov}(X) + \mathrm{Cov}(Y) = 2\mathrm{I}_d$. Remark that $\mathrm{C_p}(X)$ is, approximately, proportional to the maximal eigenvalue of $\mathrm{Cov}(X)$. Thus, for ill-conditioned covariance matrices $\frac{\mathrm{C_p}(X)}{\sigma_X^2},\frac{\mathrm{C_p}(Y)}{\sigma^2_Y}$ will not be on the same scale. It seems plausible to conjecture that the dependence on the minimal eigenvalue and Poicnar\'e constant could be replaced by a quantity which would take into consideration all eigenvalues.
\\
\\
Some other known stability results, both for log-concave vectors and for other classes of measures, may be found in \cite{toscani2015strengthened,courtade2018quantitative,courtade2016wasserstein}. The reader is referred to  \cite[Section 2.2]{courtade2016wasserstein} for a complete discussion. Let us mention one important special case, which is relevant to our results; the so-called entropy jump, first proved for the one dimensional case by Ball, Barthe and Naor (\cite{ball2003entropy}) and then generalized by Ball and Nguyen to arbitrary dimensions in \cite{ball2012entropy}. According to the latter result, if $X$ is a log-concave and isotropic random vector, then
$$\delta_{EPI, \frac{1}{2}}(X,X) \geq \frac{1}{8\mathrm{C_p}(X)}\mathrm{D}(X||G),$$
where $\mathrm{C_p}(X)$ is the Poincar\'e constant of $X$ and $G$ is the standard Gaussian. This should be compared to both Corollary \ref{cor: EPI for xi log-concave} and Theorem \ref{thm: EPI for general log-concave}. That is, in the special case of two identical measures and $\lambda = \frac{1}{2}$, their result gives a better dependence on the Poincar\'e constant than the one afforded by our results.

Ball and Nguyen (\cite{ball2012entropy}) also give an interesting motivation for these type of inequalities: They show that if for some constant $\kappa > 0$,
$$\delta_{EPI, \frac{1}{2}}(X,X) \geq  \kappa \mathrm{D}(X||G),$$
then the density $f_X$ of $X$ satisfies, $f_X(0) \leq e^{\frac{2d}{\kappa}}$. The isotropic constant of $X$ is defined by $L_X := f_X(0)^{\frac{1}{d}}$, and is the main subject of the slicing conjecture, which hypothesizes that $L_X$ is uniformly bounded by a constant, independent of the dimension, for every isotropic log-concave vector $X$. Ball and Nguyen observed that using the above fact in conjunction with an entropy jump estimate gives a bound on the isotropic constant in terms of the Poincar\'e constant, and in particular the slicing conjecture is implied by the KLS conjecture.

Our final results give improved bounds under the assumption that $X$ and $Y$ are already close to being Gaussian, in terms of relative entropy, or if one them is a Gaussian. We record these results in the following theorems. 
\begin{theorem}\label{thm: low entropy}
	Suppose that $X, Y$ be isotropic log-concave vectors such that $\mathrm{C_p}(X),\mathrm{C_p}(Y) \leq \mathrm{C_p}$ for some $\mathrm{C_p} < \infty$. Suppose further that $\mathrm{D}(X||G), \mathrm{D}(Y||G) \leq \frac{1}{4}$, then
	$$\delta_{EPI, \lambda}(X,Y) \geq \frac{\lambda(1-\lambda)}{36\mathrm{C_p}}\left(\mathrm{D}(X||G) + \mathrm{D}(Y||G)\right)$$
\end{theorem}
The following gives an improved bound in the case that one of the random vectors is a Gaussian, and holds in full generality with respect to the other vector, without a log-concavity assumption.
\begin{theorem} \label{thm: gaussian conv}
	Let $X$ be a centered random vector with finite Poincar\'e constant, $\mathrm{C_p}(X) < \infty$. Then
	$$\delta_{EPI, \lambda}(X,G) \geq \left(\lambda - \frac{\lambda\left(\mathrm{C_p}(X) - 1\right) - \ln\left(\lambda\left(\mathrm{C_p}(X)-1\right)+1\right)}{\mathrm{C_p}(X)  -\ln\left(\mathrm{C_p}(X)\right)- 1}\right)\mathrm{D}(X||G).$$
\end{theorem}
\begin{remark}
	When $\mathrm{C_p}(X) \geq 1$, the following inequality holds
	$$\left(\lambda - \frac{\lambda\left(\mathrm{C_p}(X) - 1\right) - \ln\left(\lambda\left(\mathrm{C_p}(X)-1\right)+1\right)}{\mathrm{C_p}(X)  -\ln\left(\mathrm{C_p}(X)\right)- 1}\right) \geq \frac{\lambda(1-\lambda)}{\mathrm{C_p}(X)}.$$
\end{remark}
\begin{remark}
	Theorem \ref{thm: gaussian conv} was already proved in \cite{courtade2016wasserstein} by using a slightly different approach. Denote by $\mathrm{I}(X||G)$, the relative Fisher information of the random vector $X$. In \cite{fathi2016logsobolev} the authors proof the following improved log-Sobolev inequality.
	$$\mathrm{I}(X||G) \geq 2\mathrm{D}(X||G)\frac{(1-\mathrm{C_p}(X))^2}{\mathrm{C_p}(X)(\mathrm{C_p}(X)-\ln\left(\mathrm{C_p}(X)-1\right))}.$$
	The theorem follows by integrating the inequality along the Ornstein-Uhlenbeck semi-group. 
\end{remark}

\subsection*{Acknowledgments}
We are grateful to Alex Zhai for several enlightening exchanges of ideas, and are thankful to Bo'az Klartag and Max Fathi for useful discussions. We would also like to thank Tom Courtade for his thoughtful comments concerning a preliminary draft and for suggesting that we generalize the proof of Theorem \ref{thm: EPI for general log-concave} for arbitrary covariance structures.

\section{Bounding the deficit via martingale embeddings} \label{sec: approach}

Our approach is based on ideas somewhat related to the ones which appear in \cite{eldan2018clt}: the very high-level plan of the proof is to embed the variables $X,Y$ as the terminal points of some martingales and express the entropies of $X,Y$ and $X+Y$ as functions of the associates quadratic co-variation processes. One of the main benefits in using such an embedding is that the co-variation process of $X+Y$ can be easily expressed in terms on the ones of $X,Y$, as demonstrated below. In \cite{eldan2018clt} these ideas where used to produce upper bounds for the entropic central limit theorem, so it stands to reason that related methods may be useful here. It turns out, however, that in order to produce meaningful bounds for the Shannon-Stam inequality, one needs a more intricate analysis, since this inequality corresponds to a second-derivative phenomenon: whereas for the CLT one only needs to produce upper bounds on the relative entropy, here we need to be able to compare, in a non-asymptotic way, two relative entropies.
	
In particular, our martingale embedding is constructed using the entropy minimizing technique developed by F\"ollmer (\cite{follmer1985entropy, follmer1986time}) and later Lehec (\cite{lehec2013representation}). This construction has several useful features, one of which is that it allows us to express the relative entropy of a measure in $\RR^d$ in terms of a variational problem on the Wiener space. In addition, upon attaining a slightly different point of view on this process, that we introduce here, the behavior of this variational expression turns out to be tractable with respect to convolutions. 

In order to outline the argument, fix centered measures $\mu$ and $\nu$ on $\RR^d$ with finite second moment. Let $X \sim \mu$, $Y \sim \nu$ be random vectors and $G \sim \gamma$ a standard Gaussian random vector. \\ \\
\textbf{An entropy-minimizing drift}. 
Let $B_t$ be a standard Brownian motion on $\RR^d$ and denote by $\mathcal{F}_t$ its natural filtration. In the sequel, the following process plays a fundamental role:

\begin{equation} \label{eq: follmer drift}
v^X_t = \arg\min\limits_{u_t} \frac{1}{2}\int\limits_0^1\EE\left[\norm{u_t}_2^2\right]dt,
\end{equation}
where the minimum is taken with respect to all processes $u_t$ adapted to $\FF_t$, such that
$$B_1 + \int\limits_0^1u_tdt \sim \mu.$$
Amazingly, under mild assumptions on $\mu$, and in particular in the case that $\mu$ is log-concave, there exists a unique minimizer to Equation \eqref{eq: follmer drift}, from which we construct the process  
$$X_t := B_t + \int\limits_0^tv^X_sds,$$
also known as the F\"ollmer process, with $v_t^X$ being the associated F\"ollmer drift. We refer the reader to \cite{lehec2013representation} for proofs of the existence and uniqueness of the process, as well as of a few other facts summarized below.

It turns out that the process $v_t^X$ is a martingale (which goes together with the fact that it minimizes a quadratic form) which is given by the equation
\begin{equation} \label{eq: follmer equation}
v_t^X = \nabla_x \ln\left(P_{1-t}(f_X(X_t))\right),
\end{equation} 
where $f_X$ is the density of $X$ with respect to the standard Gaussian  and $P_{1-t}$ denotes the heat semi-group. In fact, Girsanov's formula gives a very useful relation between the energy of the drift and the entropy of $X$, namely,
\begin{equation} \label{eq: entropy energy}
\mathrm{D}(X||G) = \frac{1}{2}\int\limits_0^1\EE\left[\norm{v_t^X}_2^2\right]dt.
\end{equation}
This gives the following alternative interpretation for the process: suppose that the Wiener space is equipped with an underlying probability measure $P$, with respect to which the process $B_t$ is a Brownian motion as above. Let $Q$ be a measure on Wiener space such that 
$$
\frac{dP}{dQ} = \frac{d \mu}{d \gamma} (X_1),
$$
then the process $X_t$ is a Brownian motion with respect to the measure $Q$. By the representation theorem for the Brownian bridge, this tells us that the process $X_t$ conditioned on $X_1$ is a Brownian bridge between $0$ and $X_1$. In particular, we have
\begin{equation} \label{eq: Brownian bridge}
X_t \stackrel{d}{=} tX_1 +\sqrt{ t(1-t)}G.
\end{equation}
\textbf{Lehec's proof of the Shannon-Stam inequality.} For the sake of intuition, we now repeat Lehec's argument to reproduce the Shannon-Stam inequality \eqref{eq: relative shannon-stam} using this process. Let $X_t := B^X_t + \int\limits_0^tv^X_sds$ and $Y_t := B^Y_t + \int\limits_0^tv^Y_sds$ be the F\"ollmer processes associated to $X$ and $Y$, where $B_t^X$ and $B_t^Y$ are independent Brownian motions. For $\lambda \in (0,1)$, define the new processes $$w_t = \sqrt{\lambda} v_t^X + \sqrt{1-\lambda}v_t^Y,$$
and
$$\tilde{B}_t = \sqrt{\lambda}B_t^X + \sqrt{1-\lambda}B_t^Y.$$
By the independence of $B_t^X$ and $B_t^Y$, $\tilde{B}_t$ is a Brownian motion and 
$$\tilde{B}_1 + \int\limits_0^1 w_tdt = \sqrt{\lambda}X_1 + \sqrt{1-\lambda}Y_1.$$ 
Note that as the $v_t^X$ is martingale, we have for every $t\in[0,1]$,
$$\EE\left[v_t^X\right] = \EE\left[X_1\right] = 0.$$
Using equations \eqref{eq: follmer drift} and \eqref{eq: entropy energy} and recalling that the processes are independent, we finally have
\begin{align*}
\mathrm{D}(\sqrt{\lambda}X_1 + \sqrt{1-\lambda}Y_1||G) &\leq \frac{1}{2}\int\limits_0^1\EE\left[\norm{w_t}_2^2\right]dt \\
&= \frac{\lambda}{2}\int\EE\left[\norm{v_t^X}_2^2\right]dt + \frac{1-\lambda}{2}\int\EE\left[\norm{v_t^Y}_2^2\right]dt \\
&= \lambda \mathrm{D}(X_1||G) + (1-\lambda)\mathrm{D}(Y_1||G).
\end{align*}
This recovers the Shannon-Stam inequality in the form \eqref{eq: relative shannon-stam}. \\ \\
\textbf{An alternative point of view: Replacing the drift by a varying diffusion coefficient.} 
Lehec's proof gives rise to the following idea: Suppose the processes $v_t^X$ and $v_t^Y$ could be coupled in a way such that the variance of the resulting process $\sqrt{\lambda} v_t^X + \sqrt{1-\lambda}v_t^Y$ was smaller than that of $w_t$ above. Such a coupling would improve on \eqref{eq: relative shannon-stam} and that is the starting point of this work.

As it turns out, however, it is easier to get tractable bounds by working with a slightly different interpretation of the above processes, in which the role of the drift is taken by an adapted diffusion coefficient of a related process. 

The idea is as follows: Suppose that $M_t := \int\limits_0^t F_sdB_s$ is a martingale, where $F_t$ is some positive-definite matrix valued process adapted to $\FF_t$. Consider the drift defined by
\begin{equation} \label{eq: mart to drift}
u_t := \int\limits_0^t\frac{F_s - \mathrm{I}_d}{1-s}dB_s.
\end{equation}
We then claim that $B_1 + \int\limits_{0}^1u_tdt = M_1$. To show this, we use the stochastic Fubini Theorem (\cite{veraar2012stochastic}) to write
$$\int\limits_0^1F_tdB_t = \int\limits_0^1\mathrm{I}_ddB_t+\int\limits_0^1\left(F_t-\mathrm{I}_d\right)dB_t = B_1 + \int\limits_0^1\int\limits_t^1 \frac{F_t - \mathrm{I}_d}{1-t}dsdB_t = B_1 + \int\limits_0^1u_tdt.$$

Since we now expressed the random variable $M_1$ as the terminal point of a standard Brownian motion with an adapted drift, the minimality property of the F\"ollmer drift together with equation \eqref{eq: entropy energy} immediately produce a bound on its entropy. Namely, by using It\^o's isometry and Fubini's theorem we have the bound
\begin{equation} \label{eq: entropy bound}
\mathrm{D}(M_1||G) \stackrel{\eqref{eq: entropy energy}}{\leq} \frac{1}{2}\int\limits_{0}^1\EE\left[\norm{u_t}_2^2\right] =\frac{1}{2} \mathrm{Tr}\int\limits_0^1\int\limits_0^t\frac{\EE\left[\left(F_s - \mathrm{I}_d\right)^2\right]}{(1-s)^2}dsdt =\frac{1}{2} \mathrm{Tr}\int\limits_{0}^1\frac{\EE\left[\left(F_t - \mathrm{I}_d\right)^2\right]}{1-t}dt.
\end{equation}

This hints at the following possible scheme of proof: in order to give an upper bound for the expression $\mathrm{D}(\sqrt{\lambda}X_1 + \sqrt{1-\lambda}Y_1||G)$, it suffices to find martingales $M_t^X$ and $M_t^Y$ such that $M_1^{X}, M_1^Y$ have the laws of $X$ and $Y$, respectively, and such that the $\lambda$-average of the covariance processes is close to the identity. 

The F\"ollmer process gives rise to a natural martingale: Consider $\EE\left[X_1|\mathcal{F}_t\right]$, the associated Doob martingale. By the martingale representation theorem (\cite[Theorem 4.3.3]{oksendal2003stochastic}) there exists a uniquely defined adapted matrix valued process $\Gamma_t^X$, for which
\begin{equation}\label{eq:defgamma}
\EE\left[X_1|\mathcal{F}_t\right] = \int\limits_0^t\Gamma^X_sdB_s^X.
\end{equation}
By following the construction in \eqref{eq: mart to drift} and considering the process $\tilde{v}_t^X := \int\limits_0^t\frac{\Gamma^X_s -\mathrm{I}_d}{1-s}dB^X_s$, it is immediate that $B_1 + \int\limits_{0}^1\tilde{v}_t^Xdt = X_1$. Observe that $v_t - \tilde{v}_t$ is a martingale and that for every $t \in [0,1]$, $\int\limits_t^1(v_s^X - \tilde{v}_s^X)ds|\FF_t = 0,$ almost surely. It thus follows that $v_t^X$ and $\tilde v_t^X$ are almost surely the same process. We conclude the following representation for the F\"ollmer drift,
\begin{equation} \label{eq: follmer to gamma}
v_t^X = \int\limits_0^t\frac{\Gamma^X_s -\mathrm{I}_d}{1-s}dB^X_s.
\end{equation}
The matrix $\Gamma_t^X$ turns out to be positive definite almost surely, (in fact, it has an explicit simple representation, see Proposition \ref{prop: gamma properties} below), which yields, by the combining \eqref{eq: entropy energy} with same calculation as in \eqref{eq: entropy bound},
\begin{equation} \label{eq: martingale entropy}
\mathrm{D}(X||G) = \frac{1}{2}\int\limits_0^1\frac{\mathrm{Tr}\left(\EE\left[\left(\Gamma_s^X-\mathrm{I}_d \right)^2\right]\right)}{1-t}dt.
\end{equation}

Given the processes $\Gamma_t^X$ and $\Gamma_t^Y$, we are now in position to express $\sqrt{\lambda} X + \sqrt{1-\lambda} Y$ as the terminal point of a martingale, towards using \eqref{eq: entropy bound}, which would lead to a bound on $\delta_{EPI,\lambda}$. We define $$\tilde{\Gamma_t} := \sqrt{\lambda\left(\Gamma_t^X\right)^2 + (1-\lambda)\left(\Gamma_t^Y\right)^2},$$ and a martingale $\tilde{B}_t$ which satisfies $$\tilde{B}_0 = 0 \text{ and } d\tilde{B}_t = \tilde{\Gamma}^{-1}_t\left(\sqrt{\lambda}\Gamma_t^XdB_t^X + \sqrt{1-\lambda}\Gamma_t^YdB_t^Y\right).$$
Since $\Gamma_t^X$ and $\Gamma_t^Y$ are invertible almost surely and independent, it holds that
$$[\tilde{B}]_t = t\mathrm{I}_d,$$
where $[\tilde{B}]_t$ denotes the quadratic co-variation of $\tilde{B}_t$. Thus, by Levy's characterization, $\tilde{B}_t$ is a standard Brownian motion and 
we have the following equality in law
$$\int\limits_0^1\tilde{\Gamma_t}d\tilde{B}_t = \sqrt{\lambda}\int\limits_0^1\Gamma_t^XdB_t^X + \sqrt{1-\lambda}\int\limits_{0}^1\Gamma_t^YdB_t^Y \stackrel{d}{=} \sqrt{\lambda}X_1 + \sqrt{1-\lambda}Y_1.$$
We can now invoke \eqref{eq: entropy bound} to get
$$\mathrm{D}\left(\sqrt{\lambda}X_1 + \sqrt{1-\lambda}Y_1\big|\big|G\right) \leq \frac{1}{2} \int\limits_{0}^1\frac{\mathrm{Tr}\left(\EE\left[\left(\tilde{\Gamma_t} - \mathrm{I}_d\right)^2\right]\right)}{1-t}dt.$$
Combining this with the identity \eqref{eq: martingale entropy} finally gives a bound on the deficit in the Shannon-Stam inequality, in the form
\begin{align} \label{eq:jump inequality}
\delta_{EPI,\lambda}(X,Y) &\geq \frac{1}{2}\int\limits_0^1\frac{\mathrm{Tr}\left(\lambda\EE\left[\left(\Gamma_t^X - \mathrm{I}_d\right)^2\right] +(1-\lambda)\EE\left[\left(\Gamma_t^Y - \mathrm{I}_d\right)^2\right] - \EE\left[\left(\tilde{\Gamma_t} - \mathrm{I}_d\right)^2\right]\right)}{1-t}dt \nonumber \\
&= \int\limits_0^1\frac{\mathrm{Tr}\left(\EE\left[\tilde{\Gamma_t}\right] - \lambda\EE\left[\Gamma^X_t\right] - (1-\lambda)\EE\left[\Gamma_t^Y\right]\right)}{1-t}dt.
\end{align}
The following technical lemma will allow us to give a lower bound for the right hand side in terms of the variances of the processes $\Gamma_t^X, \Gamma_t^Y$. Its proof is postponed to the end of the section.
\begin{lemma} \label{lem:matrices}
	Let $A$ and $B$ be positive definite matrices and denote 
	$$(A,B)_\lambda : = \lambda A+(1-\lambda)B \text{ and } (A^2,B^2)_\lambda : = \lambda A^2+(1-\lambda)B^2.$$
	Then 
	$$\mathrm{Tr}\left(\sqrt{(A^2,B^2)_\lambda} - (A,B)_\lambda\right) = \lambda(1-\lambda)\mathrm{Tr}\left(\left(A- B\right)^2\left(\sqrt{(A^2,B^2)_\lambda} + (A,B)_\lambda\right)^{-1}\right).$$
\end{lemma}

Combining the lemma with the estimate obtained in \eqref{eq:jump inequality} produces the following result, which will be our main tool in studying $\delta_{EPI, \lambda}$.
\begin{lemma} \label{lem: jump bound}
	Let $X$ and $Y$ be centered random vectors on $\RR^d$ with finite second moment, and let $\Gamma_t^X, \Gamma_t^Y$ be defined as above. Then,
	\begin{align}
	&\delta_{EPI,\lambda}(X,Y)\geq \nonumber \\
	&\lambda(1-\lambda)\int\limits_{0}^1\frac{\mathrm{Tr}\left(\EE\left[\left(\Gamma_t^X - \Gamma_t^Y\right)^2\left(\sqrt{\lambda\left(\Gamma_t^X\right)^2 + (1-\lambda)\left(\Gamma_t^Y\right)^2} + \lambda\Gamma_t^X + (1-\lambda)\Gamma_t^Y\right)^{-1}\right]\right)}{1-t}dt. \label{eq:mainepibound}
	\end{align}
\end{lemma}
The expression on the right-hand side of \eqref{eq:mainepibound} may seem unwieldy, however, in many cases it can be simplified. For example, if it can be shown that, almost surely, $\Gamma_t^X, \Gamma_t^Y \preceq c_t\mathrm{I}_d$ for some deterministic $c_t > 0$, then we obtain the more tractable inequality
\begin{equation} \label{eq: jump bound}
	\delta_{EPI,\lambda}(X,Y)\geq
	\frac{\lambda(1-\lambda)}{2}\int\limits_{0}^1\frac{\mathrm{Tr}\left(\EE\left[\left(\Gamma_t^X - \Gamma_t^Y\right)^2\right]\right)}{(1-t)c_t}dt.
\end{equation}
As we will show, this is the case when the random vectors are log-concave.

\begin{proof}[Proof of Lemma \ref{lem:matrices}]
	We have
	\begin{align*}
	&\mathrm{Tr}\left(\sqrt{(A^2,B^2)_\lambda} - (A,B)_\lambda\right)\\
	&= \mathrm{Tr}\left(\left(\sqrt{(A^2,B^2)_\lambda} - (A,B)_\lambda\right)\left(\sqrt{(A^2,B^2)_\lambda} + (A,B)_\lambda\right)\left(\sqrt{(A^2,B^2)_\lambda} + (A,B)_\lambda\right)^{-1}\right).\\
	\end{align*}
	As
	\begin{align*}
	&\left(\sqrt{(A^2,B^2)_\lambda} - (A,B)_\lambda\right)\left(\sqrt{(A^2,B^2)_\lambda} + (A,B)_\lambda\right)\\ = \lambda(1-\lambda)&\left(A^2 + B^2- AB - BA\right) +\sqrt{(A^2,B^2)_\lambda}(A,B)_\lambda - (A,B)_\lambda \sqrt{(A^2,B^2)_\lambda},
	\end{align*}
	we have the equality
	\begin{align*}
	\mathrm{Tr}\left(\sqrt{(A^2,B^2)_\lambda} - (A,B)_\lambda\right) &= \lambda(1-\lambda)\mathrm{Tr}\left(\left(A^2 + B^2- \left(AB + BA\right)\right)\left(\sqrt{(A^2,B^2)_\lambda} + (A,B)_\lambda\right)^{-1}\right)\\
	&+\mathrm{Tr}\left(\sqrt{(A^2,B^2)_\lambda}(A,B)_\lambda\left(\sqrt{(A^2,B^2)_\lambda} + (A,B)_\lambda\right)^{-1}\right)\\
	&-\mathrm{Tr}\left((A,B)_\lambda\sqrt{(A^2,B^2)_\lambda}\left(\sqrt{(A^2,B^2)_\lambda} + (A,B)_\lambda\right)^{-1}\right)
	\end{align*}
	Finally, as the trace is invariant under any permutation of three symmetric matrices we have that 
	$$\mathrm{Tr}\left(AB\left(\sqrt{(A^2,B^2)_\lambda} + (A,B)_\lambda\right)^{-1}\right) = \mathrm{Tr}\left(BA\left(\sqrt{(A^2,B^2)_\lambda} + (A,B)_\lambda\right)^{-1}\right),$$
	and
	\begin{align*}
	\mathrm{Tr}\Big(\sqrt{(A^2,B^2)_\lambda}(A,B)_\lambda&\left(\sqrt{(A^2,B^2)_\lambda} + (A,B)_\lambda\right)^{-1}\Big)\\
	&=\mathrm{Tr}\left((A,B)_\lambda\sqrt{(A^2,B^2)_\lambda}\left(\sqrt{(A^2,B^2)_\lambda} + (A,B)_\lambda\right)^{-1}\right).
	\end{align*}
	Thus,
	\begin{align*}
	\mathrm{Tr}\left(\sqrt{(A^2,B^2)_\lambda} - (A,B)_\lambda\right)
	=\lambda(1-\lambda)\mathrm{Tr}\left(\left(\left(A - B\right)^2\right)\left(\sqrt{(A^2,B^2)_\lambda} + (A,B)_\lambda\right)^{-1}\right),
	\end{align*}
	as required.
\end{proof}

	\subsection{The F\"ollmer process associated to log-concave random vectors}
	In this section, we collect several results pertaining to the F\"ollmer process. Throughout the section, we fix a random vector $X$ in $\RR^n$ and associate to it the F\"ollmer process $X_t$, defined in the previous section, as well as the process $\Gamma^X_t$, defined in equation \eqref{eq:defgamma} above. The next result lists some of its basic properties, and we refer to \cite{eldan2018clt,ElDuke18} for proofs.
	\begin{proposition} \label{prop: gamma properties}
		For $t \in (0,1)$ define 
		$$f^t_X(x) := f_X(x)\exp\left(\frac{\norm{x-X_t}_2^2}{2(1-t)}\right)Z_{t,X}^{-1},$$
		where $f_X$ is the density of $X$ with respect to the standard Gaussian and $Z_{t,X}$ is a normalizing constant defined so that $\int\limits_{\RR^d} f_X^t = 1$. Then
		\begin{itemize}
			\item $f_X^t$ is the density of the random measure $\mu_t := X_1|\mathcal{F}_t$ with respect to the standard Gaussian and $\Gamma^X_t = \frac{\mathrm{Cov}\left(\mu_t\right)}{1-t}$.
			\item $\Gamma^X_t$ is almost surely a positive definite matrix, in particular, it is invertible.
			\item For all $t \in (0,1)$, we have
			\begin{equation}\label{eq:degamme}
			\frac{d}{dt}\EE\left[\Gamma^X_t\right] = \frac{\EE\left[\Gamma^X_t\right] - \EE\left[\left(\Gamma^X_t\right)^2\right]}{1-t}.
			\end{equation}
			\item The following identity holds
			\begin{equation}\label{eq:idvtgamma}
			\EE\left[v_t^X\otimes v_t^X\right] = \frac{\mathrm{I}_d - \EE\left[\Gamma^X_t\right]}{1-t} + \mathrm{Cov}(X) - \mathrm{I}_d,
			\end{equation} 
			for all $t \in [0,1]$. In particular, if $\mathrm{Cov}(X) \preceq \mathrm{I}_d$, then $\EE\left[\Gamma^X_t\right] \preceq \mathrm{I}_d$.
		\end{itemize}
		
	\end{proposition}
	
	In what follows, we restrict ourselves to the case that $X$ is log-concave. Using this assumption we will establish several important properties for the matrix $\Gamma_t$. For simplicity, we will write $\Gamma_t := \Gamma_t^X$ and $v_t := v_t^X$.
	The next result shows that the matrix $\Gamma_t$ is bounded almost surely.
	\begin{lemma} \label{lem: bounded gamma}
		Suppose that $X$ is log-concave, then for every $t \in (0,1)$
		$$\Gamma_t \preceq \frac{1}{t}\mathrm{I}_d.$$
		Moreover, if for some $\xi >0$, $X$ is $\xi$-uniformly log-concave then
		$$\Gamma_t \preceq \frac{1}{(1-t)\xi + t}\mathrm{I}_d.$$
	\end{lemma}
	\begin{proof}
			By Proposition \ref{prop: gamma properties}, $\mu_t$, the law of $X_1|\mathcal{F}_t$ has a density $\rho_t$, with respect to the Lebesgue measure, proportional to $$f_X(x) \exp\left(\frac{\norm{x}_2^2}{2}\right)\exp\left(-\frac{\norm{x-X_t}_2^2}{2(1-t)}\right) = f_X(x) \exp\left(\frac{\norm{x}_2^2(1-t) - \norm{x-X_t}_2^2}{2(1-t)}\right).$$
			Consequently, since $-\nabla^2f_X \succeq 0$,
			$$-\nabla^2\ln\left(\rho_t\right) = -\nabla^2f_X  - \left(1 - \frac{1}{1-t}\right)\mathrm{I}_d \succeq \frac{t}{1-t}\mathrm{I}_d.$$
			 It follows that, almost surely, $\mu_t$ is $\frac{t}{1-t}$-uniformly log-concave. According to the Brascamp-Lieb inequality (\cite{brascamp1976extensions}) $\alpha$-uniform log-concavity implies a spectral gap of $\alpha$, and in particular $\textrm{Cov}(\mu_t) \preceq \frac{1 - t}{t}\mathrm{I}_d$ and so, $\Gamma_t = \frac{\mathrm{Cov}(\mu_t)}{1-t} \preceq \frac{1}{t}\mathrm{I}_d$. If, in addition, $X$ is $\xi$-uniformly log-concave, so that $-\nabla^2f_X \succeq \xi \mathrm{I}_d$, then we may write
			 $$-\nabla^2\ln(\rho_t) \succeq \left(\xi + \frac{t}{1-t}\right)\mathrm{I}_d = \frac{(1-t)\xi +t}{(1-t)}\mathrm{I}_d  $$ 
			 and the arguments given above show
			$\textrm{Cov}(\mu_t) \preceq \frac{(1-t)}{(1-t)\xi+ t}\mathrm{I}_d$.
			Thus,
			$$\Gamma_t \preceq \frac{1}{(1-t)\xi + t}\mathrm{I}_d.$$
	\end{proof}
Our next goal is to use the formulas given in the above lemma in order to bound from below the expectation of $\Gamma_t$. We begin with a simple corollary.
\begin{corollary} \label{corr: bounded uniform log-concave}
	Suppose that $X$ is $1$-uniformly log-concave, then for every $t \in [0,1]$
	$$\EE\left[\Gamma_t\right] \succeq \mathrm{Cov}(X).$$
\end{corollary}
\begin{proof}
	By \eqref{eq:degamme}, we have 
	$$\frac{d}{dt}\EE\left[\Gamma_t\right] = \frac{\EE\left[\Gamma_t\right] - \EE\left[\Gamma_t^2\right]}{1-t}.$$
	By Lemma \ref{lem: bounded gamma}, $\Gamma_t\preceq \mathrm{I}_d$, which shows
	$$\frac{d}{dt}\EE\left[\Gamma_t\right] \succeq 0.$$ 
	Thus, for every $t$,
	$$\EE\left[\Gamma_t\right] \succeq \EE\left[\Gamma_0\right] = \mathrm{Cov}(X|\mathcal{F}_0) = \mathrm{Cov}(X).$$
\end{proof}
To produce similar bounds for general log-concave random vectors, we require more intricate arguments. Recall that $\mathrm{C_p}(X)$ denotes the Poincar\'e constant of $X$. 
\begin{lemma} \label{lem: follmer poincare}
	If $X$ is centered and has a finite a Poincar\'e constant $\mathrm{C_p}(X) < \infty$, then
	$$\EE\left[v_t^{\otimes 2}\right] \preceq \left(t^2\mathrm{C_p}(X)+t(1-t)\right)\frac{d}{dt}\EE\left[v_t^{\otimes 2}\right].$$
\end{lemma}
\begin{proof}
	Recall that, by equation \eqref{eq: Brownian bridge}, we know that $X_t$ has the same law as $tX_1 + \sqrt{t(1-t)}G$, where $G$ is a standard Gaussian independent of $X_1$.
	Since $\mathrm{C_p}(tX) = t^2\mathrm{C_p}(X) $ and since the Poincar\'e constant is sub-additive with respect to convolution (\cite{courtade2018bounds}) we get
	$$\mathrm{C_p}(X_t) \leq t^2\mathrm{C_p}(X) + t(1-t).$$
	The drift, $v_t$, is a function of $X_t$ and $\EE\left[v_t\right] = 0$. Equation \eqref{eq: follmer equation} implies that $\nabla_x v_t(X_t)$ is a symmetric matrix, hence the Poincar\'e inequality yields
	$$\EE\left[v_t^{\otimes 2}\right] \preceq \left(t^2\mathrm{C_p}(X)+t(1-t)\right)\EE\left[\nabla_x v_t(X_t)^2\right].$$
	As $v_t(X_t)$ is a martingale, by It\^o's lemma we have
	$$dv_t(X_t) = \nabla_xv_t(X_t)dB_t.$$
	An application of It\^o's isometry then shows
	$$\EE\left[\nabla_x v_t(X_t)^2\right] = \frac{d}{dt}\EE\left[v_t(X_t)^{\otimes 2}\right],$$
	where we have again used the fact that $\nabla_x v_t(X_t)$ is symmetric.
\end{proof}
Using the last lemma, we can deduce lower bounds on the matrix $\Gamma_t^X$ in terms of the Poincar\'e constant.
\begin{corollary} \label{cor: lower gamma bound}
	Suppose that $X$ is log-concave and that $\sigma^2$ is the minimal eigenvalue of $\mathrm{Cov}(X)$. Then,
	\begin{itemize}
		\item For every $t \in \left[0,\frac{1}{ 2\frac{\mathrm{C_p}(X)}{\sigma^2}+1}\right]$, 
		$\EE\left[\Gamma_t\right] \succeq \frac{\min(1,\sigma^2)}{3}\mathrm{I}_d.$
		\item For every $t \in \left[ \frac{1}{ 2\frac{\mathrm{C_p}(X)}{\sigma^2}+1}, 1\right]$, 
		$\EE\left[\Gamma_t\right] \succeq \frac{\min(1,\sigma^2)}{3}\frac{1}{t\left(2\frac{\mathrm{C_p}(X)}{\sigma^2}+1\right)}\mathrm{I}_d$.
	\end{itemize}
\end{corollary}
\begin{proof}
	Using Equation \eqref{eq: follmer to gamma}, It\^o's isometry and the fact that $\Gamma_t$ is symmetric, we deduce that
	$$\frac{d}{dt}\EE\left[v_t^{\otimes 2}\right] = \EE\left[\left(\frac{\Gamma_t - \mathrm{I}_d}{1-t}\right)^2\right],$$
	Combining this with equation \eqref{eq:idvtgamma} and using Lemma \ref{lem: follmer poincare}, we get
	\begin{equation} \label{eq: poincare inequality}
	\mathrm{Cov}(X) - \mathrm{I}_d + \frac{\mathrm{I}_d - \EE\left[\Gamma_t\right]}{1-t} \preceq \left(t^2\mathrm{C_p}(X)+t(1-t)\right)\frac{\EE\left[\Gamma_t^2\right]-2\EE\left[\Gamma_t\right]+\mathrm{I}_d}{(1-t)^2}.
	\end{equation}
	In the case where $X$ is log-concave, by Lemma \ref{lem: bounded gamma}, $\Gamma_t \preceq\frac{1}{t}\mathrm{I}_d$ almost surely, therefore $\EE\left[\Gamma_t^2\right] \preceq \frac{1}{t}\EE\left[\Gamma_t\right]$. The above inequality then becomes
	\begin{align*}
	(1-t)^2\left(\sigma^2 - 1\right)\mathrm{I}_d &+ (1-t)(\mathrm{I}_d - \EE\left[\Gamma_t\right])\\ 
	&\preceq \left(t\mathrm{C_p}(X)+(1-t)\right)\EE\left[\Gamma_t\right]+\left(t^2\mathrm{C_p}(X)+t(1-t)\right)\left(\mathrm{I}_d-2\EE\left[\Gamma_t\right]\right).
	\end{align*}
	Rearranging the inequality shows 
	$$\frac{\sigma^2 - 2t\sigma^2 - \mathrm{C_p}(X)t^2\ + t^2\sigma^2}{2-4t-2\mathrm{C_p}(X)t^2 + \mathrm{C_p}(X)t +2t^2}\mathrm{I}_d \preceq \EE\left[\Gamma_t\right].$$
	As long as $t \leq \frac{1}{ 2\left(\frac{\mathrm{C_p}(X)}{\sigma^2}\right)+1}$, we have 
	\begin{align*}
	&\text{if } \sigma^2\geq 1, \ \ \ \ \frac{1}{3}\mathrm{I}_d\preceq\frac{\sigma^2\left(4\mathrm{C_p}(X) - \sigma^2\right)}{2\mathrm{C_p}(X)(\sigma^2 + 4)-\sigma^4}\mathrm{I}_d \preceq \EE\left[\Gamma_t\right],\\
	&\text{if } \sigma^2< 1, \ \ \ \frac{\sigma^2}{3}\mathrm{I}_d\preceq\frac{\sigma^2\left(4\mathrm{C_p}(X) - \sigma^2\right)}{2\mathrm{C_p}(X)(\sigma^2 + 4)-\sigma^4}\mathrm{I}_d \preceq \EE\left[\Gamma_t\right],
	\end{align*}
	which gives the first bound.
	By \eqref{eq:defgamma}, we also have the bound
	$$\frac{d}{dt}\EE\left[\Gamma_t\right] = \frac{\EE\left[\Gamma_t\right] - \EE\left[\Gamma_t^2\right]}{1-t}\succeq\frac{1 - \frac{1}{t}}{1-t}\EE\left[\Gamma_t\right] = -\frac{1}{t}\EE\left[\Gamma_t\right].$$
	The differential equation 
	$$g'(t) = -\frac{g(t)}{t}, g\left(\frac{1}{ 2\frac{\mathrm{C_p}(X)}{\sigma^2}+1}\right) = \frac{\min(1,\sigma^2)}{3}$$
	has a unique solution given by
	$$g(t) = \frac{\min(1,\sigma^2)}{3}\frac{1}{ t\left(2 \frac{\mathrm{C_p}(X)}{\sigma^2}+1\right)}.$$
	Using Gromwall's inequality, we conclude that for every $t \in \left[\frac{1}{ 2\frac{\mathrm{C_p}(X)}{\sigma^2}+1},1\right]$,
	$$\EE\left[\Gamma_t\right] \succeq \frac{\min(1,\sigma^2)}{3}\frac{1}{ t\left(2\frac{\mathrm{C_p}(X)}{\sigma^2}+1\right)}\mathrm{I}_d.$$
\end{proof}
We conclude this section with a comparison lemma that will allow to control the values of $\EE\left[\norm{v_t}_2^2\right]$.
	\begin{lemma} \label{lem: gromwall poincare}
		Let $t_0 \in [0,1]$ and suppose that $X$ is centered with a finite Poincar\'e constant $\mathrm{C_p}(X) < \infty$. Then
		\begin{itemize}
			\item For $t_0 \leq t \leq 1,$
			$$\EE\left[\norm{v_t}_2^2\right] \geq  \EE\left[\norm{v_{t_0}}_2^2\right]\frac{t_0\left(\mathrm{C_p}(X)-1\right)t + t}{t_0\left(\mathrm{C_p}(X)-1\right)t + t_0}.$$
			\item For $0 \leq t \leq t_0,$
			$$\EE\left[\norm{v_t}_2^2\right] \leq  \EE\left[\norm{v_{t_0}}_2^2\right]\frac{t_0\left(\mathrm{C_p}(X)-1\right)t + t}{t_0\left(\mathrm{C_p}(X)-1\right)t + t_0}.$$
		\end{itemize}
	\end{lemma}
	\begin{proof}
		Consider the differential equation
		$$g(t) = \left(\mathrm{C_p}(X)t^2 + t(1-t)\right)g'(t) \text{ with initial condition } g(t_0) =  \EE\left[\norm{v_{t_0}}_2^2\right].$$
		It has a unique solution given by
		$$g(t) = \EE\left[\norm{v_{t_0}}_2^2\right]\frac{t_0\left(\mathrm{C_p}(X)-1\right)t + t}{t_0\left(\mathrm{C_p}(X)-1\right)t + t_0}.$$
		The bounds follow by applying Gromwall's inequality combined with the result of Lemma \ref{lem: follmer poincare}.
	\end{proof}

\section{Stability for $1$-uniformly log-concave random vectors} \label{sec: uniformly log-concave}
In this section, we assume that $X$ and $Y$ are both $1$-uniformly log-concave. Let $B_t^X, B_t^Y$ be independent standard Brownian motions and consider the associated processes $\Gamma_t^X, \Gamma_t^Y$ defined as in Section \ref{sec: approach}.

The key fact that makes the uniform log-concave case easier is Lemma \ref{lem: bounded gamma}, which implies that $\Gamma_t^X,\Gamma_t^Y \preceq \mathrm{I}_d$ almost surely. In this case,  Lemma \ref{lem: jump bound} simplifies to
\begin{equation} \label{eq: jump as vars}
\delta_{EPI, \lambda}(X,Y) \geq \frac{\lambda(1-\lambda)}{2}\int\limits_{0}^1\left(\frac{\mathrm{Tr}\left(\mathrm{Var}(\Gamma_t^X)\right)}{1-t} + \frac{\mathrm{Tr}\left(\mathrm{Var}(\Gamma_t^Y)\right)}{1-t} + \frac{\mathrm{Tr}\left(\left(\EE\left[\Gamma_t^X\right] - \EE\left[\Gamma_t^Y\right]\right)^2\right)}{1-t}\right)dt,
\end{equation}
where we have used the fact that
$$\mathrm{Tr}\left(\EE\left[\left(\Gamma_t^X - \Gamma_t^Y\right)^2\right]\right) = \mathrm{Tr}\left(\EE\left[\left(\Gamma_t^X - \EE\left[\Gamma_t^X\right]\right)^2\right] + \EE\left[\left(\Gamma_t^Y - \EE\left[\Gamma_t^Y\right]\right)^2\right] + \left(\EE\left[\Gamma_t^X\right] - \EE\left[\Gamma_t^Y\right]\right)^2\right).$$
Consider the two Gaussian random vectors defined as
$$G_X = \int\limits_{0}^1\EE\left[\Gamma_t^X\right]dB_t^X \text{ and } G_Y = \int\limits_0^1\EE\left[\Gamma_t^Y\right]dB_t^Y,$$
and observe that
$$X = \int\limits_{0}^1\Gamma_t^XdB_t^X = \int\limits_0^1\left(\Gamma_t^X - \EE\left[\Gamma_t^X\right]\right)dB_t^X + \int\limits_{0}^1\EE\left[\Gamma_t^X\right]dB_t^X =\int\limits_0^1\left(\Gamma_t^X - \EE\left[\Gamma_t^X\right]\right)dB_t^X +G_X.$$
This induces a coupling between $X$ and $G_X$ from which we obtain, using It\^o's Isometry,
$$\mathcal{W}_2^2 \left(X, G_X\right) \leq \EE\left[\left(\int\limits_0^1\left(\Gamma_t^X - \EE\left[\Gamma_t^X\right]\right)dB_t^X\right)^2 \right] = \int\limits_0^1\mathrm{Tr}\left(\mathrm{Var}\left(\Gamma_t^X\right)\right)dt,$$
and an analogous estimate also holds for $Y$. We may now use $\EE\left[\Gamma_t^X\right]$ and $\EE\left[\Gamma_t^Y\right]$ as the diffusion coefficients for the same Brownian motion to establish
$$\mathcal{W}^2_2(G_X, G_Y) \leq \EE\left[\left(\int\limits_{0}^1\left(\EE\left[\Gamma_t^X\right]-\EE\left[\Gamma_t^Y\right]\right)dB_t\right)^2\right]=\int\limits_0^1\mathrm{Tr}\left(\left(\EE\left[\Gamma_t^X\right]- \EE\left[\Gamma_t^Y\right]\right)^2\right)dt.$$ 
Plugging these estimates into \eqref{eq: jump as vars} reproves the following bound, which is identical to Theorem 1 in \cite{courtade2016wasserstein}.
\begin{theorem}
	Let $X$ and $Y$ be $1$-uniformly log-concave centered vectors and let $G_X, G_Y$ be defined as above. Then,
\begin{align*}
		\delta_{EPI, \lambda}(X,Y) \geq
		\frac{\lambda(1-\lambda)}{2}\left(\mathcal{W}_2^2\left(X||G_X\right) +\mathcal{W}_2^2\left(Y||G_Y\right)+ \mathcal{W}_2^2\left(G_X,G_Y\right)\right).
		\end{align*}
\end{theorem} 
To obtain a bound for the relative entropy towards the proof of Theorem \ref{thm: stability for uniform}, we will require a slightly more general version of inequality \eqref{eq: entropy bound}. This is the content of the next lemma, whose proof is similar to the argument presented above. The main difference comes from applying Girsanov's theorem to a re-scaled Brownian motion, from which we obtain an expression analogous to \eqref{eq: entropy energy}. The reader is referred to \cite[Lemma 2]{eldan2018clt}, for a complete proof.
\begin{lemma} \label{lem: entropy bound}
	Let $F_t$ and $E_t$ be two $F_t$-adapted matrix-valued processes and let $X_t$, $M_t$ be two processes defined by
	$$Z_t = \int\limits_{0}^tF_sdB_s, \text{ and } M_t = \int\limits_{0}^tE_sdB_s.$$
	Suppose that for every $t\in[0,1]$, $E_t \succeq c\mathrm{I}_d$ for some deterministic $c > 0$, then
	$$\mathrm{D}(Z_1||M_1) \leq \mathrm{Tr}\int\limits_0^1\frac{\EE\left[\left(F_t - E_t\right)^2\right]}{c^2(1-t)}dt.$$
\end{lemma}
\begin{proof}[Proof of Theorem \ref{thm: stability for uniform}]
By Corollary \ref{corr: bounded uniform log-concave}
$$\EE\left[\Gamma_t^X\right] \succeq \sigma_X\mathrm{I}_d \text{ and } \EE\left[\Gamma_t^Y\right] \succeq \sigma_Y\mathrm{I}_d\text{ for every } t\in [0,1].$$ 
We invoke Lemma \ref{lem: entropy bound} with $E_t = \EE\left[\Gamma_t^X\right]$ and $F_t = \Gamma_t^X$ to obtain
$$\sigma_X^2\mathrm{D}(X||G_X)\leq \int\limits_{0}^1\frac{\mathrm{Tr}\left(\mathrm{Var}\left(\Gamma_t^X\right)\right)}{1-t}dt.$$
Repeating the same argument for $Y$ gives
$$\sigma_Y^2\mathrm{D}(Y||G_Y)\leq \int\limits_{0}^1\frac{\mathrm{Tr}\left(\mathrm{Var}\left(\Gamma_t^Y\right)\right)}{1-t}dt.$$
By invoking Lemma \ref{lem: entropy bound} with $F_t = \EE\left[\Gamma_t^X\right]$ and $E_t = \EE\left[\Gamma_t^Y\right]$ and then one more time after switching between $F_t$ and $E_t$, and summing the results, we get
$$\frac{\sigma_Y^2}{2}\mathrm{D}(G_X||G_Y) + \frac{\sigma_X^2}{2}\mathrm{D}(G_Y||G_X)\leq \int\limits_{0}^1\frac{\mathrm{Tr}\left(\left(\EE\left[\Gamma_t^X\right] - \EE\left[\Gamma_t^Y\right]\right)^2\right)}{1-t}dt.$$
Plugging the above inequalities into \eqref{eq: jump as vars} concludes the proof.
\end{proof}
\section{Stability for general log-concave random vectors} \label{sec: log-concave}
Fix $X, Y$, centered log-concave random vectors in $\RR^d$, such that \begin{equation} \label{eq: assumption}
\mathrm{Cov}(Y) + \mathrm{Cov}(X) = 2\mathrm{I}_d,
\end{equation} with $\sigma_X^2,\sigma_Y^2$  the corresponding minimal eigenvalues of $\mathrm{Cov}(X)$ and $\mathrm{Cov}(Y)$. Assume further that $\frac{\mathrm{C_p}(Y)}{\sigma_Y^2},\frac{\mathrm{C_p}(X)}{\sigma_X^2} \leq \mathrm{C_p}$, for some $\mathrm{C_p} >1$. Again, let $B_t^X$ and $B_t^Y$ be independent Brownian motions and consider the associated processes $\Gamma_t^X, \Gamma_t^Y$ defined as in Section \ref{sec: approach}.

The general log-concave case, in comparison with the case where $X$ and $Y$ are uniformly log-concave, gives rise to two essential difficulties. Recall that the results in the previous section used the fact that an upper bound for the matrices $\Gamma_t^X,\Gamma_t^Y$, combined with equation \eqref{eq:mainepibound} gives the simpler bound \eqref{eq: jump as vars}. Unfortunately, in the general log-concave case, there is no upper bound uniform in $t$, which creates the first problem. The second issue has to do with the lack of respective lower bounds for $\EE[\Gamma_t^X]$ and $\EE[\Gamma_t^Y]$: in view of Lemma \ref{lem: entropy bound}, one needs such bounds in order to obtain estimates on the entropies.

The solution of the second issue lies in Corollary \ref{cor: lower gamma bound}, which gives a lower bound for the processes in terms on the Poincar\'e constants. We denote $\xi = \frac{1}{(2\mathrm{C_p}+1)}\frac{\min(\sigma_Y^2,\sigma_X^2)}{3}$, so that the corollary gives
\begin{equation} \label{eq: general lower bound}
\EE\left[\Gamma_t^Y\right],\EE\left[\Gamma_t^X\right] \succeq \xi\mathrm{I}_d.
\end{equation}

Thus, we are left with the issue arising from the lack of a uniform upper bound for the matrices $\Gamma_t^X,\Gamma_t^Y$. Note that Lemma \ref{lem: bounded gamma} gives $\Gamma_t^X\preceq \frac{1}{t}\mathrm{I}_d$, a bound which is not uniform in $t$. To illustrate how one may overcome this issue, suppose that there exists an $\eps>0$, such that
$$\int\limits_0^\eps\frac{\mathrm{Tr}\left(\EE\left[\left(\Gamma_t^X - \Gamma_t^Y\right)^2\right]\right)}{(1-t)}dt < \frac{1}{2} \int\limits_0^1\frac{\mathrm{Tr}\left(\EE\left[\left(\Gamma_t^X - \Gamma_t^Y\right)^2\right]\right)}{(1-t)}dt.$$
In such a case, Lemma \ref{lem: jump bound} would imply
$$\delta_{EPI, \lambda}(X,Y) \gtrsim \frac{\lambda(1-\lambda)}{\eps}\mathrm{Tr}\int\limits_{0}^1\frac{\EE\left[\left(\Gamma_t^X-\Gamma_t^Y\right)^2\right]}{1-t}dt.$$
Towards finding an $\eps$ such that the above holds, note that since $v_t^X$ is a martingale, and using \eqref{eq: entropy energy} we have for every $t_0 \in [0,1],$
\begin{equation} \label{eq: martingale bound}
(1-t_0)\mathrm{D}\left(X||G\right) =\frac{1-t_0}{2}\int\limits_{0}^1\EE\left[\norm{v_t^X}_2^2\right]dt \leq \frac{1}{2}\int\limits_{t_0}^1\EE\left[\norm{v_t^X}_2^2\right]dt \leq \mathrm{D}\left(X||G\right).
\end{equation}
Observe that
$$\mathrm{Tr}\left(\EE\left[\left(\Gamma_t^X - \Gamma_t^Y\right)^2\right]\right) = \mathrm{Tr}\left(\EE\left[\left(\Gamma_t^X - \mathrm{I}_d\right)^2\right]+ \EE\left[\left(\Gamma_t^Y - \mathrm{I}_d\right)^2\right] -2\EE\left[\mathrm{I}_d - \Gamma_t^X\right]\EE\left[\mathrm{I}_d - \Gamma_t^Y\right]\right).$$
Using the relation in \eqref{eq: follmer to gamma}, Fubini's theorem shows
\begin{align*}
\int\limits_{t_0}^{1}\EE\left[\norm{v_t^X}_2^2\right]dt &= \mathrm{Tr}\int\limits_{t_0}^{1}\int\limits_{0}^t \frac{\EE\left[\left(\Gamma_s^X-\mathrm{I}_d\right)^2\right]}{(1-s)^2}dsdt \nonumber\\
&= \mathrm{Tr}\int\limits_{0}^{t_0}\int\limits_{t_0}^{1}\frac{\EE\left[\left(\Gamma_s^X-\mathrm{I}_d\right)^2\right]}{(1-s)^2}dtds + \mathrm{Tr}\int\limits_{t_0}^{1}\int\limits_{s}^1\frac{\EE\left[\left(\Gamma_s^X-\mathrm{I}_d\right)^2\right]}{(1-s)^2}dtds\nonumber\\
&=  (1-t_0)\EE\left[\norm{v_{t_0}^X}_2^2\right] + \mathrm{Tr}\int\limits_{t_0}^{1}\frac{\EE\left[\left(\Gamma_s^X-\mathrm{I}_d\right)^2\right]}{1-s}ds.
\end{align*}
Combining the last two displays gives
\begin{align}\label{eq: truncation}
\mathrm{Tr}\int\limits_{t_0}^1\frac{\EE\left[\left(\Gamma_t^X - \Gamma_t^Y\right)^2\right]}{1-t}dt = \int\limits_{t_0}^1\bigg(\EE\left[\norm{v_t^X}_2^2\right]&+ \EE\left[\norm{v_t^Y}_2^2\right]\bigg)dt - (1-t_0)\left(\EE\left[\norm{v_{t_0}^X}_2^2\right] +\EE\left[\norm{v_{t_0}^Y}_2^2\right]\right)\nonumber\\
&-2\mathrm{Tr}\int\limits_{t_0}^1\frac{\EE\left[\mathrm{I}_d - \Gamma_t^X\right]\EE\left[\mathrm{I}_d - \Gamma_t^Y\right]}{1-t}dt.
\end{align} 
Using \eqref{eq:idvtgamma}, we have the identities:
$$\frac{\EE\left[\mathrm{I}_d-\Gamma_t^X\right]}{1-t} = \EE\left[v^X_t\otimes v^X_t\right] + \mathrm{I}_d-\mathrm{Cov}(X)$$
and
$$\frac{\EE\left[\mathrm{I}_d-\Gamma_t^Y\right]}{1-t} = \EE\left[v^Y_t\otimes v^Y_t\right] + \mathrm{I}_d-\mathrm{Cov}(Y),$$
from which we deduce
\begin{align*}
2\frac{\EE\left[\mathrm{I}_d-\Gamma_t^X\right]\EE\left[\mathrm{I}_d-\Gamma_t^Y\right]}{1-t} =& \left(\mathrm{I}_d- \EE\left[\Gamma_t^Y\right]\right)\EE\left[v^X_t\otimes v^X_t\right] + \left(\mathrm{I}_d- \EE\left[\Gamma_t^X\right]\right)\EE\left[v^Y_t\otimes v^Y_t\right]\\
&+\left(\mathrm{I}_d- \EE\left[\Gamma_t^Y\right]\right)\left(\mathrm{I}_d - \mathrm{Cov}(X)\right) + \left(\mathrm{I}_d- \EE\left[\Gamma_t^X\right]\right)\left(\mathrm{I}_d - \mathrm{Cov}(Y)\right).
\end{align*}
Let $\{w_i\}_{i=1}^d$ be an orthornormal basis of eigenvectors corresponding to the eigenvalues $\{\lambda_i\}_{i=1}^d$ of $\mathrm{I}_d-\EE\left[\Gamma_t^X\right]$. The following observation, which follows from the above identities, is crucial: if $\lambda_i \leq 0$ then necessarily $\langle w_i, \mathrm{Cov}(X)w_i\rangle \geq 1$. In this case, by assumption \eqref{eq: assumption}, $\langle w_i, \mathrm{Cov}(Y)w_i\rangle \leq 1$ and 
$$\left  \langle w_i,\frac{\EE\left[\mathrm{I}_d-\Gamma_t^X\right]\EE\left[\mathrm{I}_d-\Gamma_t^Y\right]}{1-t} w_i \right  \rangle\leq 0. $$
Our aim is to bound \eqref{eq: truncation} from below; thus, in the calculation of the trace in the RHS, we may disregard all $w_i$ corresponding to negative $\lambda_i$. Moreover, if $\lambda_i \geq 0$, we need only consider the cases where 
$$\langle w_i,\left(\mathrm{I}_d- \EE\left[\Gamma_t^Y\right]\right)w_i\rangle \geq 0,$$
as well.
Since,
\begin{align*}
2 \left \langle w_i,\frac{\EE\left[\mathrm{I}_d-\Gamma_t^X\right]\EE\left[\mathrm{I}_d-\Gamma_t^Y\right]}{1-t} w_i \right \rangle =& \langle w_i,\EE\left[\mathrm{I}_d-\Gamma_t^X\right] w_i\rangle\left(\EE\left[\langle v_t^Y,w_i\rangle^2\right] + 1 - \langle w_i, \mathrm{Cov}(Y)w_i\rangle\right)\\
&+\langle w_i,\EE\left[\mathrm{I}_d-\Gamma_t^Y\right] w_i\rangle\left(\EE\left[\langle v_t^X,w_i\rangle^2\right] + 1 - \langle w_i, \mathrm{Cov}(X)w_i\rangle\right),
\end{align*}
under the assumptions taken on $w_i$, we see that all the terms are positive. Using the estimate \eqref{eq: general lower bound}, the previous equation is bounded from above by
\begin{align*}
(1-\xi)\big(&\EE\left[\langle v_t^Y,w_i\rangle^2\right] + 1 - \langle w_i, \mathrm{Cov}(Y)w_i\rangle + \EE\left[\langle v_t^X,w_i\rangle^2\right] + 1 - \langle w_i, \mathrm{Cov}(X)w_i\rangle\big) \\
&=(1-\xi)\big(\EE\left[\langle v_t^Y,w_i\rangle^2\right] + \EE\left[\langle v_t^X,w_i\rangle^2\right]\big),
\end{align*}
where we have used \eqref{eq: assumption}. Summing over all the relevant $w_i$ we get
$$2\mathrm{Tr}\frac{\EE\left[\mathrm{I}_d-\Gamma_t^X\right]\EE\left[\mathrm{I}_d-\Gamma_t^Y\right]}{1-t} \leq (1-\xi)\left(\EE\left[\norm{v^X_t}_2^2\right] + \EE\left[\norm{v^Y_t}_2^2\right]\right).$$
Plugging this into \eqref{eq: truncation} and using \eqref{eq: martingale bound} we have thus shown
\begin{align} \label{eq:remainder bound}
\mathrm{Tr}\int\limits_{t_0}^1\frac{\EE\left[\left(\Gamma_t^X - \Gamma_t^Y\right)^2\right]}{1-t}dt\geq 2\xi(1&-t_0) \left(\mathrm{D}(X||G) + \mathrm{D}(Y||G)\right) \nonumber\\
&- (1-t_0)\left(\EE\left[\norm{v_{t_0}^X}_2^2\right] +\EE\left[\norm{v_{t_0}^Y}_2^2\right]\right).
\end{align}
This suggests that it may be useful to bound $\EE\left[\norm{v^X_{t_0}}_2^2\right]$ from above, for small values of $t_0$, which is the objective of the next lemma.
\begin{lemma} \label{lem: smallv_t}
	If $X$ is centered and has a finite Poincar\'e constant $\mathrm{C_p}(X) < \infty$, then for every $s \leq \frac{1}{3(2\mathrm{C_p}(X)+1)}$ the following holds
	$$\EE\left[\norm{v^X_{s^2}}^2_2\right] < \frac{s}{4}\cdot \mathrm{D}(X||G).$$ 
\end{lemma}
\begin{proof}
	Suppose to the contrary that $\EE\left[\norm{v_{s^2}^X}^2_2\right] \geq \frac{s}{4}\cdot \mathrm{D}(X||G)$. 
	Invoking Lemma \ref{lem: gromwall poincare} with $t_0 = s^2$ gives
	$$\EE\left[\norm{v_t^X}_2^2\right] \geq \mathrm{D}(X||G)\cdot \frac{t\left((\mathrm{C_p}(X) -1)s^2 + 1\right)}{4\left((\mathrm{C_p}(X) -1)st + s\right)},$$
	whenever $t \geq s^2$. 
	Thus,
	\begin{align} \label{eq: gronwall bound}
	\int\limits_{s^2}^1\EE\left[\norm{v_t^X}_2^2\right]dt &\geq \mathrm{D}(X||G)\int\limits_{s^2}^1\frac{t\left((\mathrm{C_p}(X) -1)s^2 + 1\right)}{4\left((\mathrm{C_p}(X) -1)st + s\right)}dt \nonumber\\
	&= \mathrm{D}(X||G)\left((\mathrm{C_p}(X)-1)s^2+1\right)\frac{(\mathrm{C_p}(X)-1)t-\ln\left(t\left(\mathrm{C_p}(X)-1\right) + 1\right)}{4(\mathrm{C_p}(X)-1)^2s}\Bigg\vert_{s^2}^1.
	\end{align}
	Note now that for $s\leq \frac{1}{3(2\mathrm{C_p}(X)+1)}$
	$$\frac{d}{ds} \frac{t\left((\mathrm{C_p}(X) -1)s^2 + 1\right)}{4\left((\mathrm{C_p}(X) -1)st + s\right)} = \frac{\left(\mathrm{C_p}(X)-1\right)s^2t-1}{s^2((\mathrm{C_p}(X)-1)t+1)}<0,$$
	and in particular we may substitute $s = \frac{1}{3(2\mathrm{C_p}(X)+1)}$ in \eqref{eq: gronwall bound}. 
	In this case, a straightforward calculation yields
	$$\int\limits_{\xi^2_X}^1\EE\left[\norm{v_t^X}_2^2\right]dt > \mathrm{D}(X||G),$$
	which contradicts the identity \eqref{eq: entropy energy}, and concludes the proof by contradiction.
\end{proof}
We would like to use the lemma with the choice $s = \xi^2$. In order to verify the condition on the lemma which amounts to $\xi^2 \leq \frac{1}{3(2\mathrm{C_p}(X)+1)}$, we first remark that if $\sigma_X^2 \leq 1$, then it is clear that $\xi \leq \frac{1}{3(2\mathrm{C_p}(X)+1)}$. Otherwise, $\sigma_X^2 \geq 1$ and 
$$\xi \leq \frac{1}{2\frac{\mathrm{C_p}(X)}{\sigma_X^2}+1}\frac{\sigma_Y^2}{3} \leq \frac{1}{2\frac{\mathrm{C_p}(X)}{\sigma_X^2}+1}\frac{2 - \sigma_X^2}{3} \leq\frac{1}{3(2\mathrm{C_p}(X)+1)}.$$
As the same reasoning is also true for $Y$, we now choose $t_0 = \xi^2$, which allows to invoke the previous lemma in \eqref{eq:remainder bound} and to establish:
\begin{equation}\label{eq: partial variance}
\mathrm{Tr}\int\limits_{\xi^2}^1\frac{\EE\left[\left(\Gamma_t^X - \Gamma_t^Y\right)^2\right]}{1-t}dt \geq \xi\left(\mathrm{D}(X||G) +\mathrm{D}(Y||G)\right). 
\end{equation}
We are finally ready to prove the main theorem.
\begin{proof}[Proof of Theorem \ref{thm: EPI for general log-concave}]
	
	Denote $\xi = \frac{1}{(2\mathrm{C_p}+1)}\frac{\min(\sigma_Y^2,\sigma_X^2)}{3}$. 
	Since $X$ and $Y$ are log-concave, by Lemma \ref{lem: bounded gamma}, $\Gamma_t^X, \Gamma_t^Y \preceq \frac{1}{t}\mathrm{I}_d$ almost surely. Thus, Lemma \ref{lem: jump bound} gives
	\begin{align*}
	\delta_{EPI, \lambda}(X,Y) \geq \frac{\xi^2\lambda(1-\lambda)}{2}\int\limits_{\xi^2}^1 \frac{\mathrm{Tr}\left(\EE\left[(\Gamma_t^X - \Gamma_t^Y)^2\right]\right)}{1-t}dt.
	\end{align*}
	By noting that $\mathrm{C_p} \geq 1$, the bound \eqref{eq: partial variance} gives
	\begin{align*}
	\delta_{EPI, \lambda}(X,Y)&\geq \frac{\xi^3\lambda(1-\lambda)}{2}\left(\mathrm{D}\left(X||G\right) +\mathrm{D}\left(Y||G\right)\right)\\
	&\geq K\lambda(1-\lambda)\left(\frac{\min(\sigma_Y^2,\sigma_X^2)}{\mathrm{C_p}}\right)^3\left(\mathrm{D}\left(X||G\right) + \mathrm{D}\left(Y||G\right) 
	\right),
	\end{align*}
	for some numerical constant $K>0$.
\end{proof}
\section{Further results} \label{sec: extensions}
\subsection{Stability for low entropy log concave measures}
In this section we focus on the case where $X$ and $Y$ are log-concave and isotropic. Similar to the previous section, we set $\xi_X = \frac{1}{3(2\mathrm{C_p}(X) + 1)}$, so that by Corollary \ref{cor: lower gamma bound},
$$\EE\left[\Gamma_t^X\right] \succeq \xi_X\mathrm{I}_d.$$  Towards the proof of Theorem \ref{thm: low entropy}, we first need an analogue of Lemma \ref{lem: smallv_t}, for which we sketch the proof here.
\begin{lemma} \label{lem: smaller v_t}
	If $X$ is centred and has a finite Poincar\'e constant $\mathrm{C_p}(X) < \infty$,
	$$\EE\left[\norm{v_{\xi_X}}^2_2\right] < \frac{1}{4}\mathrm{D}(X||G).$$ 
\end{lemma} 
\begin{proof}
	Assume by contradiction that $\EE\left[\norm{v_{\xi_X}}_2^2\right] \geq \frac{1}{4}\mathrm{D}(X||G)$. In this case, Lemma \ref{lem: gromwall poincare} implies, for every $t \geq \xi_X$,
	$$\EE\left[\norm{v_t^X}_2^2\right] \geq \mathrm{D}(X||G)\cdot \frac{t\left((\mathrm{C_p}(X) -1)\xi_X + 1\right)}{4\left((\mathrm{C_p}(X) -1)\xi_Xt + \xi_X\right)}.$$
	A calculation then shows that
	$$\int\limits_{\xi_X}^1\EE\left[\norm{v^X_t}_2^2\right]dt \geq \mathrm{D}(X||G),$$
	which is a contradiction to \eqref{eq: entropy energy}.
\end{proof}

\begin{proof} [Proof of Theorem \ref{thm: low entropy}]
	Since $v_t^X$ is a martingale, $\EE\left[\norm{v_t^X}_2^2\right]$ is an increasing function. By \eqref{eq: entropy energy} we deduce the elementary inequality
	$$\EE\left[\norm{v_s^X}_2^2\right] \leq \frac{1}{1-s}\int\limits_0^1\EE\left[\norm{v_t^X}^2_2\right]dt = \frac{2\mathrm{D}(X||G)}{1-s},$$
	which holds for every $s \in[0,1]$. For isotropic $X$, Equation \eqref{eq:idvtgamma} shows that, for all $t \in [0,1]$,
	$$(1-t)\EE\left[\norm{v^X_t}_2^2\right]=\mathrm{Tr}\left(\mathrm{I}_d - \EE\left[\Gamma^X_t\right]\right) \leq 2\mathrm{D}(X||G) \leq \frac{1}{2},$$
	where the second inequality is by assumption. Note that Equation \eqref{eq:idvtgamma} also shows that $\EE\left[\Gamma_t^X\right] \preceq \mathrm{I}_d$ which yields, for every $t \in[0,1]$
	$$0 \preceq\mathrm{I}_d - \EE\left[\Gamma^X_t\right] \preceq \frac{1}{2}\mathrm{I}_d.$$
	Applying this to $Y$ as well produces the bound 
	\begin{align*}
	2\mathrm{Tr}\frac{\EE\left[\mathrm{I}_d-\Gamma_t^X\right]\EE\left[\mathrm{I}_d-\Gamma_t^Y\right]}{1-t} &\leq\frac{1}{2}\mathrm{Tr}\left(\frac{\EE\left[\mathrm{I}_d-\Gamma_t^Y\right]}{1-t}\right) + \frac{1}{2}\mathrm{Tr}\left(\frac{\EE\left[\mathrm{I}_d-\Gamma_t^X\right]}{1-t}\right)\\
	&= \frac{1}{2}\left(\EE\left[\norm{v^X_t}_2^2\right]+\EE\left[\norm{v^Y_t}_2^2\right]\right).
	\end{align*}
	Set $\xi = \min(\xi_X,\xi_Y)$. Repeating the same calculation as in \eqref{eq: truncation} and using the above gives that
	\begin{align*} 
	\mathrm{Tr}\int\limits_{\xi}^1\frac{\EE\left[\left(\Gamma_t^X - \Gamma_t^Y\right)^2\right]}{1-t}dt\geq (1&-\xi) 	\left(\mathrm{D}(X||G) + \mathrm{D}(Y||G)\right) \nonumber\\
	&- (1-\xi)\left(\EE\left[\norm{v_{\xi}^X}_2^2\right] 	+\EE\left[\norm{v_{\xi}^Y}_2^2\right]\right).
	\end{align*}
	Lemma \ref{lem: smaller v_t} implies
	$$\mathrm{Tr}\int\limits_{\xi}^1\frac{\EE\left[\left(\Gamma_t^X - \Gamma_t^Y\right)^2\right]}{1-t}dt \geq \frac{3}{4}(1-\xi)\left(\mathrm{D}(X||G) + \mathrm{D}(Y||G)\right) \geq \frac{1}{2}\left(\mathrm{D}(X||G) + \mathrm{D}(Y||G)\right).$$
	Finally, by Lemma \ref{lem: bounded gamma}, $\Gamma_t^X,\Gamma_t^Y \preceq \frac{1}{t}\mathrm{I}_d$ almost surely for all $t \in [0,1]$. We now invoke Lemma \ref{lem: jump bound} to obtain
	\begin{align*}
	\delta_{EPI, \lambda}(X,Y) &\geq\frac{\lambda(1-\lambda)}{2\xi}\mathrm{Tr}\int\limits_{\xi}^1\frac{\EE\left[\left(\Gamma_t^X - \Gamma_t^Y\right)^2\right]}{1-t}dt\\
	&\geq \frac{\lambda(1-\lambda)}{4\xi}\left(\mathrm{D}(X||G) + \mathrm{D}(Y||G)\right).
	\end{align*}
\end{proof}
\subsection{Stability under convolution with a Gaussian}
\begin{proof}[Proof of Theorem \ref{thm: gaussian conv}]
	Fix $\lambda \in (0,1)$, by \eqref{eq: Brownian bridge} we have that 
	$$\sqrt{\lambda}\left(\sqrt{\lambda}X_1 + \sqrt{1-\lambda}G\right) \stackrel{d}{=} B_\lambda + \int\limits_0^\lambda v_t^Xdt.$$
	As the relative entropy is affine invariant, this implies
	\begin{equation}\label{eq:entcomb}
	\mathrm{D}\left(\sqrt{\lambda}\left(\sqrt{\lambda}X_1 + \sqrt{1-\lambda}G\right)\Big|\Big|\sqrt{\lambda} G\right) = \mathrm{D}\left(\sqrt{\lambda}X_1 + \sqrt{1-\lambda}G\Big|\Big|G\right) = \frac{1}{2}\int\limits_0^\lambda\EE\left[\norm{v^X_t}_2^2\right]dt.
	\end{equation}
	Lemma \ref{lem: gromwall poincare} yields,
	$$\EE\left[\norm{v_t^X}_2^2\right] \geq \EE\left[\norm{v_\lambda^X}_2^2\right]\frac{\lambda\left(\mathrm{C_p}(X)-1\right)t + t}{\lambda\left(\mathrm{C_p}(X)-1\right)t + \lambda} \text{ for } t \geq \lambda,$$
	and
	$$\EE\left[\norm{v_t^X}_2^2\right] \leq \EE\left[\norm{v_\lambda^X}_2^2\right]\frac{\lambda\left(\mathrm{C_p}(X)-1\right)t + t}{\lambda\left(\mathrm{C_p}(X)-1\right)t + \lambda} \text{ for } t \leq \lambda.$$
	Denote
	$$I_1 := \int\limits_{\lambda}^1\frac{\lambda\left(\mathrm{C_p}(X)-1\right)t + t}{\lambda\left(\mathrm{C_p}(X)-1\right)t + \lambda}dt \text{ and } I_2 := \int\limits_{0}^\lambda\frac{\lambda\left(\mathrm{C_p}(X)-1\right)t + t}{\lambda\left(\mathrm{C_p}(X)-1\right)t + \lambda}dt.$$
	A calculation shows
	$$I_1 = \frac{\left(\lambda\left(\mathrm{C_p}(X) - 1\right) + 1\right)\left((1-\lambda)\left(\mathrm{C_p}(X) - 1\right) -\ln\left(\mathrm{C_p}(X)\right) + \ln\left(\lambda\left(\mathrm{C_p}(X)-1\right)+1\right)\right)}{\lambda\left(\mathrm{C_p}(X) - 1\right)^2},$$
	as well as
	$$I_2 =  \frac{\left(\lambda\left(\mathrm{C_p}(X) - 1\right) + 1\right)\left(\lambda(\mathrm{C_p}(X) - 1) - \ln\left(\lambda\left(\mathrm{C_p}(X)-1\right)+1\right)\right)}{\lambda\left(\mathrm{C_p}(X) - 1\right)^2}.$$
	Thus, the above bounds give
	\begin{align*}
	\mathrm{D}(X||G)&=\frac{1}{2}\int\limits_{0}^1\EE\left[\norm{v_t^X}_2^2\right]dt \geq \frac{1}{2}\int\limits_{0}^\lambda\EE\left[\norm{v_t^X}_2^2\right]dt + \frac{\EE\left[\norm{v_\lambda^X}_2^2\right]}{2}I_1,
	\end{align*}
	and
	$$0 \leq \frac{1}{2}\int\limits_{0}^\lambda\EE\left[\norm{v_t^X}_2^2\right]dt \leq \frac{1}{2}I_2.$$
	Now, since the expression $\frac{\alpha}{\alpha + \beta}$ is monotone increasing with respect to $\alpha$ and decreasing with respect to $\beta$ whenever $\alpha,\beta > 0$, those two inequalities together with \eqref{eq:entcomb} imply that
	\begin{align*}
	\mathrm{D}\left(\sqrt{\lambda}X + \sqrt{1-\lambda}G\Big|\Big|G\right) ~& \leq \frac{I_2}{I_1+I_2} \mathrm{D}(X||G) \\
	& =  \frac{\lambda\left(\mathrm{C_p}(X) - 1\right) - \ln\left(\lambda\left(\mathrm{C_p}(X)-1\right)+1\right)}{\mathrm{C_p}(X)  -\ln\left(\mathrm{C_p}(X)\right)- 1}\mathrm{D}(X||G).
	\end{align*}
	Rewriting the above in terms of the deficit in the Shannon-Stam inequality, we have established
	\begin{align*}
	\delta_{EPI, \lambda}(X,G) &= \lambda\mathrm{D}(X||G) - \mathrm{D}\left(\sqrt{\lambda}X + \sqrt{1-\lambda}G\Big|\Big|G\right) \\
	&\geq \left(\lambda - \frac{\lambda\left(\mathrm{C_p}(X) - 1\right) - \ln\left(\lambda\left(\mathrm{C_p}(X)-1\right)+1\right)}{\mathrm{C_p}(X)  -\ln\left(\mathrm{C_p}(X)\right)- 1}\right)\mathrm{D}(X||G).
	\end{align*}
\end{proof}

\bibliographystyle{acm}
\bibliography{bib}
\end{document}